\documentclass[3p,twocolumn]{elsarticle}

\usepackage[english]{babel}
\usepackage{amsmath,amssymb}
\usepackage[hang,centerlast]{subfigure}
\usepackage{tikz}
\usetikzlibrary{automata,patterns,topaths,shapes,calc,arrows}

\newtheorem{theorem}{Theorem}

\newtheorem{corollary}[theorem]{Corollary}

\newdefinition{definition}[theorem]{Definition}
\newdefinition{remark}[theorem]{Remark}
\newdefinition{claim}[theorem]{Claim}

\newproof{proof}{Proof}

\newcommand{\Acc}{\ensuremath{\mathit{Acc}}}
\newcommand{\distriset}[1]{\ensuremath{\mathcal{D}(#1)}}
\newcommand{\runs}[1]{\ensuremath{\mathit{Runs}(#1)}}
\newcommand{\runsf}[1]{\ensuremath{\mathit{Runs_f}(#1)}}
\newcommand{\Tr}[1]{\ensuremath{\mathit{T}(#1)}}
\newcommand{\trace}{\ensuremath{\mathit{tr}}}
\newcommand{\observable}{\ensuremath{\mathit{Obs}}}
\newcommand{\observableleak}{\ensuremath{\mathit{Obs}_\mathit{leak}}}
\newcommand{\observableopaque}{\ensuremath{\mathit{Obs}_\mathit{opaque}}}
\newcommand{\obs}{\ensuremath{\mathcal{O}}}
\newcommand{\Act}{\ensuremath{\mathit{Act}}}
\newcommand{\fee}{\ensuremath{\varphi}}

\newcommand{\A}{\ensuremath{\mathcal{A}}}

\newcommand{\Prob}{\ensuremath{\mathbf{P}}}
\newcommand{\N}{\ensuremath{\mathbb{N}}}
\newcommand{\PrDiscl}{\ensuremath{\textrm{PD}}}
\newcommand{\Buchi}[1]{\ensuremath{\textrm{B\"uchi}(#1)}}
\newcommand{\coBuchi}[1]{\ensuremath{\textrm{co-B\"uchi}(#1)}}
\newcommand{\parity}[1]{\ensuremath{\textrm{Parity}(#1)}}
\newcommand{\reach}[1]{\ensuremath{\textrm{Reach}(#1)}}
\newcommand{\Inf}[1]{\ensuremath{\textrm{Inf}(#1)}}
\newcommand{\Appear}[1]{\ensuremath{\textrm{Appear}(#1)}}

\title{Probabilistic Opacity for Markov Decision Processes}
\author[UPMC,CNRS]{B\'eatrice B\'erard}\ead{Beatrice.Berard@lip6.fr}
\author[IST]{Krishnendu
  Chatterjee}\ead{Krishnendu.Chatterjee@ist.ac.at}
\author[UPMC,CNRS]{Nathalie Sznajder}\ead{Nathalie.Sznajder@lip6.fr}
\address[UPMC]{Sorbonne Universit\'es, UPMC Univ Paris 06, UMR 7606, LIP6, F-75005, Paris, France} 
\address[CNRS]{CNRS, UMR 7606, LIP6, F-75005, Paris, France}\address[IST]{IST Austria (Institute of Science and
  Technology, Austria)}

\begin{document}

\begin{abstract}
  Opacity is a generic security property, that has been defined on
  (non probabilistic) transition systems and later on Markov chains
  with labels. For a secret predicate, given as a subset of runs, and
  a function describing the view of an external observer, the value of
  interest for opacity is a measure of the set of runs disclosing the
  secret. We extend this definition to the richer framework of Markov
  decision processes, where non deterministic choice is combined with
  probabilistic transitions, and we study related decidability problems
  with partial or complete observation hypotheses for the
  schedulers. We prove that all questions are decidable with complete
  observation and $\omega$-regular secrets. With partial observation,
  we prove that all quantitative questions are undecidable but the
  question whether a system is almost surely non opaque becomes
  decidable for a restricted class of $\omega$-regular secrets, as well
  as for all $\omega$-regular secrets under finite-memory schedulers.
\end{abstract}

\maketitle

\section{Introduction}
Due to the tremendous increase in network communications in the last
thirty years, a large amount of work was devoted to the study of
security properties, to ensure the preservation of secret data during
these communications. \emph{Information flow} characterizes the
(possibly illegal and indirect) transmission of such data from a high
level user to a low level one. Already in the eighties, a basic
version of non-interference was defined in~\cite{goguen82}, stating
that a system is secure if high level actions cannot be detected by
low level observations. Among all the subsequent studies, opacity was
introduced in~\cite{mazare05,bryans08} as a general framework where a
wide range of security properties can be specified, for a system
interacting with a passive attacker. 
For a system $\mathcal S$, opacity is parameterized by a secret
predicate $\fee$ described as a subset of executions and an
observation function over executions. The system is opaque if, for any
secret run in $\fee$, there is another run not in $\fee$ with the same
observation. When this property is satisfied, the passive attacker
cannot learn from the observation if the execution is secret.
Ensuring opacity by controller synthesis was
further studied in~\cite{dubreil10,dubreil12} while relations with
two-player games were established in~\cite{pinchinat}. 

Deciding opacity, however, only provides a yes/no answer, but no
evaluation of the amount of information gained by a passive
attacker. Since more and more security protocols make use of
randomization to reach some security
objectives~\cite{chaum88,reiter98}, it becomes important to extend
specification frameworks in order to handle measures of information
leaks. For this reason, quantitative approaches for security
properties were already advocated in~\cite{millen87,wittbold90},
mostly based on information theory. From this point on, numerous
studies were devoted to the computation of (covert) channel capacity
in various cases (see e.g.~\cite{mantel2009}) or more generally
information leakage.

To provide quantitative measures of opacity, several definitions have
been proposed in a probabilistic
setting~\cite{lakhnech05,berard10,boreale11b,DBLP:conf/tgc/BryansKM12,berard13,DBLP:journals/tac/SabooriH14}. They
were, however, restricted to purely probabilistic models, based on
Markov chains equipped with labels, to permit observations on runs. We
show here how to extend some measures of~\cite{berard13} to Markov
decision processes (MDPs) with infinite runs. The simplest one
computes what we call here the \emph{probabilistic disclosure},
providing a probabilistic measure for the set of runs whose
observation reveals that a secret run has been executed.  With the
richer model of MDPs, where non determinism is combined with
probabilities, a scheduler can cooperate with the passive external
observer to break the system opacity.  We focus on $\omega$-regular
secrets and morphisms for the observation functions, and prove that
the probabilistic disclosure can be computed when the scheduler can
distinguish the states of the model.  The class of $\omega$-regular
languages provides a robust specification language~\cite{Thomas97},
extending classical regular languages from finite words to infinite
words.  Such $\omega$-regular languages are often needed to express
opacity in the non probabilistic as well as the probabilistic setting.
With partial observation for the schedulers, the question whether a
system is almost surely non opaque remains decidable for a restricted
class of $\omega$-regular secrets, as well as for all $\omega$-regular
secrets under finite-memory schedulers, whereas all quantitative
problems become undecidable.  Moreover, for all decidable results we
present optimal complexity results: for complete observation (where
the scheduler can distinguish states of the model) we present
polynomial-time results with respect to the size of the model, whereas
for partial observation, for all decidable results we show
EXPTIME-completeness.

We recall some definitions for probabilistic models in
Section~\ref{sec:prel}. Opacity and disclosure are defined for Markov
decision processes in Section~\ref{sec:op} and proofs for the
(un)decidability results are given in Section~\ref{sec:res}. We
conclude in Section~\ref{sec:conc}.

\section{Preliminaries}\label{sec:prel}
For a finite alphabet $Z$, we denote by $Z^*$ the set of finite words
over $Z$, by $Z^\omega$ the set of infinite words over $Z$, with
$Z^\infty = Z^* \cup Z^\omega$.

We first recall some classical notions on automata.
\subsection{Automata}

\begin{definition}
  A (deterministic) automaton is a tuple $\A=(Q,\Sigma, \delta, q_0,
  F)$, where $Q$ is a finite set of states, $\Sigma$ is an input
  alphabet, $\delta:Q\times\Sigma \rightarrow Q$ is a transition
  function, $q_0 \in Q$ is the initial state, and $F$ is either a
  subset of $Q$, or a mapping from $Q$ to a finite subset of natural
  numbers.
\end{definition}

Accepting conditions defined from $F$ will be described hereafter.

A run of the automaton $\A$ on a word $w=a_1a_2\cdots\in\Sigma^\omega$
is an infinite sequence $\rho=q_0q_1\cdots$ such that for all $i\geq
0$, $q_{i+1}=\delta(q_i,a_{i+1})$.  The accepting runs of an automaton
are defined according to the acceptance condition.  In the sequel, we
consider B\"uchi, co-B\"uchi and parity acceptance conditions.

For a run $\rho=q_0q_1\cdots\in Q^\omega$, we let $\Inf{\rho}$ be the
set of states appearing infinitely often in the sequence. When
$F\subseteq Q$, we note $\Buchi{F}=\{\rho\in Q^\omega\mid
\Inf{\rho}\cap F \neq \emptyset\}$ and $\coBuchi{F}=\{\rho\in
Q^\omega\mid \Inf{\rho}\cap F = \emptyset\}$. When $F:Q\rightarrow
\{1,\dots,k\}$, with $k\in\mathbb{N}$, the acceptance condition is a
parity condition. We note $\parity{F}=\{\rho\in Q^\omega\mid
\min\{F(q)\mid q\in\Inf{\rho}\} \textrm{ is even}\}$.  For an
acceptance condition $\Acc\in\{\Buchi{F},\coBuchi{F},\parity{F}\}$, we
say that a run $\rho$ over a word $w$ is accepting if it is in
$\Acc$. The word $w$ is then said to be accepted by $\rho$.

We denote respectively by $L_B(\A)$, $L_C(\A)$ and $L_P(\A)$ the set
of words accepted by the runs of $\A$ in $\Buchi{F}$, $\coBuchi{F}$
and $\parity{F}$. A subset $L$ of $\Sigma^\omega$ is $\omega$-regular
if there is an automaton $\A$ such that $L= L_P(\A)$.

In the sequel, we write DBA for deterministic B\"uchi automata, DCA
for deterministic co-B\"uchi automata and DPA for deterministic parity
automata, according to the choice of acceptance condition.

\subsection{Probabilistic systems}

We consider systems modeled by Markov decision processes, that
generalize Markov chains by combining non deterministic actions with
probabilistic transitions. To define opacity measures on Markov
chains, the probabilistic transitions are equipped with labels that
may be used to define an observation function on runs. In the setting
of Markov decision processes, labels are also added on the
probabilistic transitions. They may be observed by a passive attacker
while non deterministic actions are chosen by a scheduler, as
explained below.

Given a countable set $S$, a discrete distribution is a mapping $\mu :
S \rightarrow [0,1]$ such that $\sum_{s\in S}\mu(s)=1$.  The set of
all discrete distributions on $S$ is denoted by $\distriset{S}$. 
\begin{definition}[Markov Decision Process]
  A Markov decision process (MDP) is a tuple $\mathcal{A} = (Q,\Sigma,
  \Act, \Delta,q_0)$ where:
\begin{itemize}
\item $Q$ is a finite set of states,
\item $\Act$ is a finite set of actions,
\item $\Sigma$ is a finite alphabet for the labeling of transitions,
\item $\Delta: Q\times \Act\rightarrow \mathcal{D}(\Sigma\times Q)$ is
  a (partial) transition function that associates with a state and an
  action from $\Act$ a probability distribution over the possible
  transition labels and successor states,
\item $q_0$ is the initial state.
\end{itemize}
\end{definition}

Figure~\ref{fig:pomdp} shows an MDP with four actions. Actions
$\alpha_1$ and $\alpha_2$ bear two different distributions for labels
$a$ and $b$. They start either from state $q_0$ or from state $q'_0$,
and lead to either $q_1$ or $q_2$. Actions $\beta_1$ and $\beta_2$
start from $q_1$ and $q_2$ respectively and return to $q_0$ or $q'_0$
with probability $\frac12$.

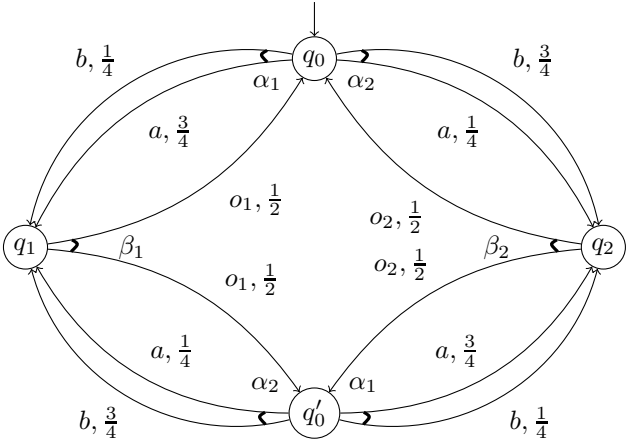
\begin{figure}[htb]
\centering
\begin{tikzpicture}[auto,node distance=4.7cm]
\tikzstyle{grouptrans}=[draw,very thick]
  \tikzstyle{every state}+=[minimum size=5pt,inner
  sep=2pt,initial text=] 
\node[state,initial above] (q0) at (0,0) {$q_0$};
\node[state, below of=q0] (q20) {$q'_0$};
\node[state] (q1) at (-3.8,-2.5) {$q_1$};
\node[state] (q2) at (3.8,-2.5) {$q_2$};

\path[->] (q0) edge [bend right=30] node[pos=0.07,anchor=south] (a1) {}
node {$a,\frac34$} (q1); 
\path[->] (q0) edge [bend right=45]
node[pos=0.07,anchor=south] (b1) {} node [swap] {$b,\frac14$} (q1);
\path[->] (q20) edge [bend left=30] node[pos=0.07,anchor=south] (a12)
{} node [swap] {$a,\frac14$} (q1); 
\path[->] (q20) edge [bend left=45]
node[pos=0.07,anchor=south] (b12) {} node {$b,\frac34$} (q1);

\node (control) at ($(barycentric cs:a1=1,b1=1) + (-0.10,-0.10)$) {};
\path[grouptrans] (a1.south) .. controls (control) .. (b1.south);

\node (control) at ($(barycentric cs:a12=1,b12=1) + (-0.10,-0.10)$) {};
\path[grouptrans] (a12.south) .. controls (control) .. (b12.south);

\path[->] (q1) edge [bend right=25] node[pos=0.07,anchor=south] (c1) {} 
node [pos=0.6,swap] {$o_1,\frac12$} (q0);
\path[->] (q1) edge [bend left=25] node[pos=0.07,anchor=south] (d1) {} 
node[pos=0.6] {$o_1,\frac12$} (q20);

\path[->] (q0) edge [bend left=30] node[pos=0.07,anchor=south] (a2) {}
node [swap] {$a,\frac14$} (q2); 
\path[->] (q0) edge [bend left=45]
node[pos=0.07,anchor=south] (b2) {} node {$b,\frac34$} (q2); 
\path[->] (q20) edge [bend right=30] node[pos=0.07,anchor=south] (a22) {} 
node {$a,\frac34$} (q2); 
\path[->] (q20) edge [bend right=45]
node[pos=0.07,anchor=south] (b22) {} node [swap] {$b,\frac14$} (q2);

\node (control) at ($(barycentric cs:a2=1,b2=1) + (0.10,-0.10)$) {};
\path[grouptrans] (a2.south) .. controls (control) .. (b2.south);

\node (control) at ($(barycentric cs:a22=1,b22=1) + (0.10,-0.10)$) {};
\path[grouptrans] (a22.south) .. controls (control) .. (b22.south);

\path[->] (q2) edge [bend left=25] node[pos=0.07,anchor=south] (c2) {}
node {$o_2,\frac12$} (q0); 
\path[->] (q2) edge [bend right=25] node[pos=0.07,anchor=south] (d2) {} 
node [swap] {$o_2,\frac12$} (q20);

\node (control) at ($(barycentric cs:c2=1,d2=1) + (-0.10,-0.10)$) {};
\path[grouptrans] (c2.south) .. controls (control) .. (d2.south);

\node (control) at ($(barycentric cs:c1=1,d1=1) + (0.10,-0.10)$) {};
\path[grouptrans] (c1.south) .. controls (control) .. (d1.south);

\node[node distance=0.4cm,below of=a1] (act1) {$\alpha_1$};
\node[node distance=0.4cm,below of=a2] (act2) {$\alpha_2$};
\node[node distance=0.4cm,above of=b22] (act12) {$\alpha_1$};
\node[node distance=0.4cm,above of=b12] (act22) {$\alpha_2$};

\node[node distance=1.4cm,right of=q1] (bet1) {$\beta_1$};
\node[node distance=1.4cm,left of=q2] (bet2) {$\beta_2$};

\end{tikzpicture}
\caption{A Markov Decision process.}
\label{fig:pomdp}
\end{figure}

The definition could be extended with an initial distribution instead
of an initial state, but we restrict to this one for the sake of
simplicity. When $\Delta(q,\alpha)$ is defined, $\alpha$ is said to be
\emph{enabled} in state $q$. Intuitively, in an execution of an MDP,
from a given state $q$, an action $\alpha\in \Act$ enabled in $q$ is
chosen non deterministically, and then the next label in $\Sigma$ and
the next state are chosen according to the probability distribution
$\Delta(q,\alpha)$.  Formally, a (finite or infinite) run of an MDP is
a sequence $\rho=q_0\cdot (\alpha_0,a_0)\cdot q_1\cdot
(\alpha_1,a_1)\cdot q_2\cdot \ldots \in Q\cdot
((\Act\times\Sigma)\cdot Q)^\infty$, also written
$q_0\xrightarrow{\alpha_0,a_0}q_1\xrightarrow{\alpha_1,a_1}q_2\ldots$
such that, for all $i\geq 0$, $\alpha_i$ is enabled in $q_i$ and
$\Delta(q_i,\alpha_i)(a_i,q_{i+1})>0$. The trace of $\rho$ is the word
$(\alpha_0,a_0)(\alpha_1,a_1)\ldots$ over $\Act\times\Sigma$ labelling
the run, obtained by projecting away the visited states. The
length of $\rho$, denoted by $|\rho|$, is the length of its trace in
$\N \cup \{\infty\}$. The set of infinite (resp. finite) runs of an
MDP $\A$ is denoted by $\runs{\A}$ (resp. $\runsf{\A}$).  The set of
traces of infinite runs of $\A$ is denoted by $\Tr{\A}$ and
$\trace:\runs{\A}\rightarrow \Tr{\A}$ is the mapping that associates
with each run its trace. For a run $\rho$, and $i < |\rho|$, we denote
by $\rho_i$ the finite run consisting of its first $i$ transitions,
and we say that $\rho_i$ is a prefix of $\rho$.

\smallskip The non determinism of MDPs is resolved by a scheduler,
that gives a probability distribution over the different actions in
$\Act$ along each finite run.
\begin{definition}[Scheduler]
  A \emph{scheduler} on $\mathcal{A} = (Q,\Sigma, \Act, \Delta,q_0)$
  is a function $\sigma : \runsf{\mathcal{A}}\rightarrow
  \mathcal{D}(\Act)$ such that, for any finite run $\rho =
  q_0\xrightarrow{\alpha_0,a_0} \ldots
  \xrightarrow{\alpha_{n-1},a_{n-1}}q_n$ of $\A$, for all
  $\alpha\in\Act$, if $\sigma(\rho)(\alpha)>0$ then $\alpha$ is
  enabled in $q_n$.
\end{definition}

A scheduler is \emph{deterministic} if $\sigma :
\runsf{\mathcal{A}}\rightarrow \Act$.  We say that a scheduler has
\emph{finite memory} if its decision
only depends on a finite set of so-called \emph{memory states}.
Similarly, a scheduler is \emph{memoryless} if its decision depends
only on the last state of the run.  Formally, they are defined as
follows.

\begin{definition}[Finite-Memory Schedulers]
  A \emph{finite-memory scheduler} on $\mathcal{A} = (Q,\Sigma, \Act,
  \Delta,q_0)$ is given by a tuple $(M,m_0,\sigma,\sigma_{up})$ where
  $M$ is a finite set of memory states, $m_0$ is the initial memory
  state, $\sigma : M\times Q\rightarrow \mathcal{D}(\Act)$ is a
  mapping such that, for all $m\in M$, for all $q\in Q$, and for all
  $\alpha\in\Act$, if $\sigma(m,q)(\alpha)>0$ then $\alpha$ is enabled
  in $q$. Finally, $\sigma_{up}:M\times \Act\times \Sigma\times
  Q\rightarrow \mathcal{D}(M)$ is the memory update function.
  
  If $|M|=1$, then the scheduler, reduced to $\sigma :
  Q\rightarrow \mathcal{D}(\Act)$ is \emph{memoryless}.
\end{definition}


\smallskip In some systems, the underlying state is only partially
observable. Those are modeled by Partially Observable MDPs, in which
some sets of states are undistinguishable for external observers
(including the scheduler):

\begin{definition}
  A partially observable Markov decision process (POMDP) is an MDP
  $\mathcal{A}=(Q,\Sigma,\Act,\Delta,q_0)$ equipped with an
  equivalence relation $\sim$ over $Q$ such that if $p \sim q$ then
  the set of actions from $\Act$ enabled in $p$ and $q$ are the same.
\end{definition}
In that case, given two sequences of states $p_0\cdots p_n$ and
$q_0\cdots q_n$, we say that $p_0\cdots p_n\sim q_0\cdots q_n$ if and
only if $p_i\sim q_i$ for all $0\leq i\leq n$. In a POMDP, the
scheduler cannot distinguish between equivalent states. The scheduler
definition is then modified:

\begin{definition}
  Let $\mathcal{A}=(Q,\Sigma,\Act,\Delta,q_0)$ be a POMDP with
  equivalence relation $\sim\subseteq Q\times Q$. An observation-based
  scheduler (or $\sim$-scheduler) is a scheduler $\sigma :
  \runsf{\mathcal{A}}\rightarrow \mathcal{D}(\Act)$ such that for any
  two finite runs $\rho= q_0\xrightarrow{\alpha_0,a_0} \ldots
  \xrightarrow{\alpha_{n-1},a_{n-1}}q_n$ and $\rho' =
  p_0\xrightarrow{\alpha_0,b_0} \ldots
  \xrightarrow{\alpha_{n-1},b_{n-1}}p_n$ with same length, if
  $p_0\cdots p_n\sim q_0\cdots q_n$, then
  $\sigma(\rho)=\sigma(\rho')$.
\end{definition}

For instance, associating with the MDP of Figure~\ref{fig:pomdp} the
three equivalence classes $\{q_0, q'_0\}$, $\{q_1\}$ and $\{q_2\}$,
produces a POMDP. In this case, the scheduler cannot know if it is in
$q_0$ or in $q'_0$ when it chooses action $\alpha_1$ or $\alpha_2$.

\smallskip Recall that, given a POMDP $\A$ and a scheduler $\sigma$, a
probability measure $\Prob_\sigma$ can be defined on
$\runs{\A}$\cite{BillingsleyMeasure95}: first it is defined on
\emph{cones}, where the cone $C_\rho$ associated with a finite run
$\rho$ is the subset of infinite runs in $\runs{\A}$ having $\rho$ as
prefix; and then it is extended to measurable sets of infinite
runs. If $ \rho= q_0\xrightarrow{\alpha_0,a_0} \ldots
\xrightarrow{\alpha_{n-1},a_{n-1}}q_n$, the probability of $C_\rho$ is
defined by:
\begin{align*}
\Prob_\sigma(C_\rho) &=
\sigma(\rho_0)(\alpha_0)\times \Delta(q_0, \alpha_0)(a_0,q_1)\times
\ldots\times\\
& \sigma(\rho_{n-1})(\alpha_{n-1})\times\Delta(q_{n-1},
\alpha_{n-1})(a_{n-1},q_n)
\end{align*}

\section{Opacity and disclosure}\label{sec:op}

The notion of opacity was originally defined in~\cite{bryans08} for a
(non probabilistic) transition system, with respect to some external
observation function and some predicate (the secret) on the runs of
the system. It extends trivially to probabilistic transition systems.
In this case, given an MDP $\A=(Q,\Sigma,\Act,\Delta,q_0)$, we
consider a predicate $\varphi \subseteq \runs{\A}$, given as an
$\omega$-regular language (the secret).  The observation the attacker
has of the runs of the MDP is defined by a morphism $\obs :
\runs{\mathcal{A}}\rightarrow\Gamma^\infty$ obtained from a given
application $\pi : Q\cup (\Act \times\Sigma)\rightarrow
\Gamma\cup\{\varepsilon\}$, where $\Gamma$ is a finite alphabet. The
morphism $\obs$ is the observation function, and the elements of
$\observable=\obs(\runs{\mathcal{A}})$ are the observables.
For a given run $\rho$, every run in
$\obs^{-1}(\obs(\rho))$ -- its observation class -- is
undistinguishable from $\rho$.  The predicate is \emph{opaque} in $\A$
for $\obs$ if each time a run satisfies the predicate, another run in
the same observation class does not. Formally, we let
$\overline{\varphi}=\runs{\mathcal{A}}\setminus\varphi$, and define
opacity as follows.

\begin{definition}[Opacity]
  Let $\mathcal{A}$ be an MDP, with observation function
  $\obs:\runs{\mathcal{A}}\rightarrow\observable$. A predicate
  $\varphi\subseteq\runs{\mathcal{A}}$ is said to be \emph{opaque} in
  $\mathcal{A}$ for $\obs$ if $\varphi\subseteq
  \obs^{-1}(\obs(\overline{\varphi}))$.
 \end{definition}

 Variants of opacity have been defined, by modifying the observation
 function or the predicate, or by requiring symmetry: the predicate
 $\varphi$ is \emph{symmetrically opaque} in $\mathcal{A}$ for $\obs$
 if both $\varphi$ and $\overline{\varphi}$ are opaque.

 Note that $\varphi$ is opaque if and only if for any
 $o\in\observable$, $\obs^{-1}(o)\not\subseteq \varphi$. By extension,
 we say that an observation class $\obs^{-1}(o)$, for $o \in
 \observable$, is itself \emph{opaque} if $\obs^{-1}(o)\not\subseteq
 \varphi$, and we define $\observableopaque$ as the set of
 corresponding observations, with
 $\observableleak=\observable\setminus\observableopaque=\{o\in
 \observable \mid \obs^{-1}(o)\subseteq \varphi\}$.


 \smallskip For instance, for the POMDP in Figure~\ref{fig:pomdp}
 above, we can define:
\begin{itemize}
\item an observation function $\obs$ from the projection $\pi$ such
 that $\pi(q) = \varepsilon$ for any $q \in Q$, $\pi(\alpha,o_1)=o_1$,
 $\pi(\alpha,o_2)=o_2$ and $\pi(\alpha,a)=\pi(\alpha,b)=\varepsilon$,
 for any $\alpha \in Act$, 
\item a predicate $\varphi$ as the set of all runs with trace in
  $(ab)^{\omega}$, where the labels $a$s and $b$s strictly alternate.
\end{itemize}

When a probabilistic system is non opaque, we are interested in
quantifying the security hole. One of the measures proposed
in~\cite{berard13} for Markov chains, is the probability of the set of
runs violating opacity. With this measure of non opacity, called here
\emph{Probabilistic Disclosure} and extended to MDPs with infinite
runs, it becomes possible to compare non opaque systems.  The measure,
computed in a worst case scenario, corresponds to the maximal
probability of disclosure over all possible schedulers. More
precisely:

 \begin{definition}[Probabilistic Disclosure]
   Let $\mathcal{A}$ be an MDP, with observation function
   $\obs:\runs{\mathcal{A}}\rightarrow\observable$, let
   $\varphi\subseteq\runs{\mathcal{A}}$ be a predicate and let
   $\sigma$ be a scheduler. The probabilistic disclosure of $\varphi$
   in $\A$ scheduled by $\sigma$ is: 

\begin{align*}\PrDiscl_\sigma(\varphi,\A, \obs)&=
   \Prob_\sigma(\varphi \setminus
   \obs^{-1}(\obs(\overline{\varphi}))) \\
   &= \sum_{o \in \observableleak}
 \Prob_\sigma(\obs^{-1}(o)).
 \end{align*}
 
 The probabilistic disclosure of $\varphi$ in $\A$ is
 $\PrDiscl(\varphi,\A, \obs)=\sup_\sigma \{\PrDiscl_\sigma(\varphi,\A,
 \obs)\}$.
 \end{definition}

 \begin{remark}
   Note that the probabilistic disclosure is well defined, since, when
   $\varphi$ is $\omega$-regular, and $\obs$ is a morphism as assumed
   above, the set of runs
   $\varphi\setminus\obs^{-1}(\obs(\overline{\varphi}))$ is
   measurable. Indeed, the class of $\omega$-regular languages is
   closed by complement, intersection, morphism and inverse morphism.
   Hence, the set
   $\varphi\setminus\obs^{-1}(\obs(\overline{\varphi}))$ is
   $\omega$-regular, thus measurable \cite{Vardi85}.
 \end{remark}

 \smallskip Questions we aim to address are the following:
\begin{enumerate}
\item The value problem: What is the value of the probabilistic
  disclosure of the system?
\item The general disclosure problem:\\ Is the value of the
  probabilistic disclosure of the system greater than some given
  threshold (\textit{i.e.} for $\delta\in [0,1]$,
  $\PrDiscl(\varphi,\A,\obs)>\delta$)?
\item The almost-sure opacity problem:\\ Is the system almost surely opaque\\
(i.e. $\PrDiscl(\varphi,\A,\obs)=0$)? 
\item The limit disclosure problem:\\ Is
  $\PrDiscl(\varphi,\A,\obs)=1$?
\item The almost-sure disclosure problem: Does there exist a scheduler
  $\sigma$ such that $\PrDiscl_\sigma(\varphi,\A,\obs)=1$?
\end{enumerate}
All these problems can be considered with a restriction to
finite-memory schedulers.  The last three questions refer to
qualitative aspects of the problem, while the two first ones concern
quantitative properties.  In the next section, we show that recent
results on MDPs (with partial or perfect observation) allow us to
answer such questions on probabilistic disclosure of the systems. More
precisely, we prove that all these questions are decidable under
perfect observation, while they are undecidable under partial
observation. However, we identify restrictions that allow to decide
the last problem.

\section{Results}\label{sec:res}    

\subsection{MDPs and Schedulers with Perfect Observation}

\begin{theorem}
  Given an MDP $\A$, an $\omega$-regular secret $\varphi$ given as a
  DPA (deterministic parity automaton), and observation function $\obs$ as a 
  morphism, the value is computable, in polynomial time in the size of
  $\A$, and exponential in the size of $\varphi$.
\end{theorem}

\begin{proof}
  Immediate, since $\varphi \setminus
  \obs^{-1}(\obs(\overline{\varphi}))$ is $\omega$-regular and can be described
  as a DPA, and from the results of~\cite{courcoubetis95,CJH04,CH11} 
  for solving MDPs with parity conditions.  \qed\end{proof}

From this theorem, it follows that:
\begin{corollary}
  The general disclosure, the limit disclosure problem, and the
  almost-sure opacity problem are decidable.
\end{corollary}

Moreover, since it is sufficient to consider memoryless deterministic
schedulers for MDPs with parity conditions~\cite{CJH04},
$\sup_\sigma\{\PrDiscl_\sigma(\varphi,\A,\obs) \}=1$ if and only if
there exists a memoryless scheduler $\sigma$ such that
$\PrDiscl_\sigma(\varphi,\A,\obs)=1$. The following result is then
obtained.

\begin{corollary}
The almost-sure disclosure problem is decidable.
\end{corollary}

Note that this result can be applied to symmetrical opacity. It can
also be extended to the case considered in~\cite{berard13} with an
observation function $\obs$ (not necessarily a morphism) producing a
finite number of observation classes such that for each $o \in
\observable$, $\obs^{-1} (o)$ is $\omega$-regular.

\subsection{POMDPs and Observation-based Schedulers}

\begin{theorem}\label{th:undecidable}
  Given a POMDP $\A$, and a morphism $\obs$ for the observation function, 
  \begin{enumerate}
  \item the almost-sure disclosure problem is undecidable for secrets given 
  as DCA, DPA.
\item the almost-sure opacity problem is undecidable for secrets given
  as DBA, DPA.
\item the limit disclosure problem, the general disclosure problem,
  and the value problem are undecidable, for secrets given as DBA,
  DCA, DPA, both with general and finite memory schedulers.
  \end{enumerate}
\end{theorem}

\begin{proof}

\begin{figure}[htb]
\centering
\begin{tikzpicture}[auto,node distance=3cm]
  \tikzstyle{every state}+=[shape=ellipse,minimum size=5pt,inner
  sep=2pt,initial text=] 
  \tikzstyle{automate}=[inner sep=10pt, draw,rectangle,rounded corners=5pt,minimum height=35pt,minimum width=45pt]
\tikzstyle{grouptrans}=[draw,very thick]
  \node[state,initial above] (q0) at (0,1.75) {$q_\iota$}; 
  \node[automate, label= left: $\A_1$] (q1) at (-2.,0) {\begin{tikzpicture}[node distance=1cm]\node[minimum size = 1pt,inner sep=0.5pt](center){};\node[dashed,inner sep=5pt,state, below right of =center]{$F_1$};\end{tikzpicture}}; 
  \node[automate, label = right: $\A_2$] (q2) at (2.,0) {\begin{tikzpicture}[node distance=1cm]\node[minimum size = 1pt,inner sep=0.5pt](center){};\node[dashed,inner sep=5pt,state, below left of =center]{$F_2$};\end{tikzpicture}};

\path[->] (q0) edge [bend left=20,pos=0.4] node[pos=0.2,anchor=south]
(a1) {} node [swap] {$a_1,\frac12$} (q1); 

\path[->] (q0) edge [bend right=20,pos=0.4]
node[pos=0.2,anchor=south] (a2) {} node {$a_2,\frac12$} (q2);

\node (control) at ($(barycentric cs:a1=1,a2=1) + (0,-0.2)$) {};
\node (lab) at (0, 1) {$\alpha_\iota$};
\path[grouptrans] (a1.south) .. controls (control) .. (a2.south);

\end{tikzpicture}
\caption{MDP $\A'$ from two copies $\A_1$ and $\A_2$ of $\A$.}
\label{fig:undec}
\end{figure}
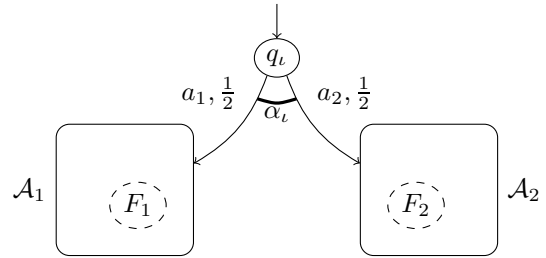

We describe a reduction from qualitative problems on POMDP to the
opacity problems addressed in this paper.  Let
$\A=(Q,\Sigma,\Act,\Delta,q_0)$ be a POMDP, with equivalence relation
$\sim$ on states.  Given a set of accepting states $F\subseteq Q$, we
let $\Acc(F)$ be either $\Buchi{F}$, or $\coBuchi{F}$ (for the
underlying non probabilistic runs of $\A$). We build a POMDP
$\A'=(Q',\Act',\Sigma',\Delta',q_\iota)$, observation function
$\obs:\runs{\A'}\rightarrow \observable$, and an $\omega$-regular
secret $\varphi$ such that schedulers for $\A$ and $\A'$ are in
correspondance (explained in more details below).
  
The POMDP $\A'$ is obtained as follows: we consider two copies $\A_1$
and $\A_2$ of $\A$ with the same alphabets $\Act$ and $\Sigma$,
denoting their disjoint set of states by $Q_1$ and $Q_2$, their
initial states by $q_0^1$ and $q_0^2$ and their target states by $F_1$ and
$F_2$, respectively. We add a new state $q_\iota$ not in $Q_1 \cup
Q_2$, a new action $\alpha_\iota$ not in $\Act$ and two new letters $a_1$
and $a_2$ not in $\Sigma$, for which the transition function is
defined by $\Delta'(q_\iota, \alpha_\iota)(a_1,q_0^1) = \Delta'(q_\iota,
\alpha_\iota)(a_2,q_0^2) = 1/2$, as depicted in Figure~\ref{fig:undec}.  The
equivalence relation on states $\sim'$ is given by $q\sim' q'$ if
$q,q'\in Q_i$ and $q\sim_i q'$ for $i=1,2$, or $q\in Q_1,q'\in Q_2$
are the copies of the same state in $Q$.  The secret $\varphi$ is the
union of two sets of runs: those starting with $q_\iota
\xrightarrow{\alpha_\iota,a_1} q_0^1$ meeting the acceptance condition
$\Acc(F_1)$ (through $\A_1$) and all the runs starting with $q_\iota
\xrightarrow{\alpha_\iota,a_2} q_0^2$ (going into $\A_2$). Formally:
\begin{displaymath}
  \varphi= \bigl (q_\iota\cdot (\alpha_\iota,a_1)\cdot \Acc(F_1)\bigr )
\cup \bigl (q_\iota\cdot(\alpha_\iota,a_2)\cdot\runs{\A_2}\bigr )
\end{displaymath}

Then, $\varphi$ can be easily given by an automaton whose acceptance
condition depends on $\Acc(F_1)$.

Finally, we define the observation function as follows: for $i=1,2$,
for all $\rho_i\in\runs{\A_i}$,
\begin{displaymath}
\obs(q_\iota\cdot(\alpha_\iota,a_i)\cdot\rho_i)=\rho,
\end{displaymath}
where $\rho$ is the corresponding run in $\A$.

Given a $\sim'$-scheduler $\sigma':\runsf{\A'}\rightarrow
\mathcal{D}(\Act)$, the probabilistic disclosure is thus:
\begin{align*}
&\PrDiscl_{\sigma'}(\varphi,\A', \obs)=\Prob_{\sigma'}(\varphi\setminus\obs^{-1}(
\obs(\overline{\varphi})))
=\\
&\Prob_{\sigma'}\Bigl( \bigl (q_\iota. (\alpha_\iota,a_1). \Acc(F_1)\bigr ) 
\cup \bigl (q_\iota.(\alpha_\iota,a_2).\Acc(F_2)\bigr )\Bigr).
\end{align*}

Since $\sigma'$ is a $\sim'$-scheduler, it is easy to see that
$\Prob_{\sigma'} \bigl (q_\iota\cdot (\alpha_\iota,a_1) \cdot
\Acc(F_1)\bigr )=\Prob_{\sigma'}\bigl (q_\iota \cdot
(\alpha_\iota,a_2) \cdot \Acc(F_2)\bigr )$.
Hence we get that
$\PrDiscl_{\sigma'}(\varphi,\A',\obs)=2\cdot\Prob_{\sigma'} \bigl
(q_\iota\cdot (\alpha_\iota,a_1)\cdot \Acc(F_1)\bigr )$

We build the $\sim$-scheduler $\sigma$ for $\A$ as follows: for each
$\rho\in\runsf{\A}$, we let
$\overline{\rho}=q_\iota\cdot(\alpha_\iota,a_1)\cdot\rho_1$ and we define
$\sigma(\rho)=\sigma'(\overline{\rho})$.
Then for the corresponding cones, we have: $\Prob_\sigma(C_\rho)=2\cdot
\Prob_{\sigma'}(C_{\overline{\rho}})$. We deduce that
$\Prob_\sigma(\Acc(F))=2\cdot
\Prob_{\sigma'}(q_\iota\cdot(\alpha_\iota,a_1)\cdot
Acc(F_1))=\PrDiscl_{\sigma'}(\varphi,\A', \obs)$.

Conversely, given a $\sim$-scheduler
$\sigma:\runsf{\A}\rightarrow\mathcal{D}(\Act)$, we define a
$\sim'$-scheduler $\sigma':\runsf{\A'}\rightarrow\mathcal{D}(\Act)$ as
follows:
\begin{displaymath}
\sigma'(q_\iota)=(\alpha_\iota\mapsto 1)
\end{displaymath}
and, for all runs $q_\iota\cdot(\alpha_\iota,a_i)\cdot\rho_i\in\runsf{\A'}$,
for $i=1,2$,
\begin{displaymath}
\sigma'(q_\iota\cdot(\alpha_\iota,a_i)\cdot\rho_i)=\sigma(\rho).
\end{displaymath}
Since $\sigma$ is a $\sim$-scheduler, $\sigma'$ is a
$\sim'$-scheduler, and for $i=1,2$, we obtain that
$\Prob_{\sigma'}(q_\iota\cdot(\alpha_\iota,a_i)\cdot\Acc(F_i))=\frac{1}{2}\Prob_\sigma(\Acc(F))$.
Hence, $\PrDiscl_{\sigma'}(\varphi,\A',\obs)=\Prob_\sigma(\Acc(F))$.

\smallskip

 Then, there exists a
$\sim$-scheduler $\sigma$ for $\A$ such that
$\Prob_\sigma(\Acc(F))>0$ if and only if there exists a
$\sim'$-scheduler $\sigma'$ for $\A'$ such that
$\PrDiscl_{\sigma'}(\varphi,\A',\obs)>0$. 
%
Also, there exists a
$\sim$-scheduler $\sigma$ for $\A$ such that
$\Prob_\sigma(\Acc(F))=1$ if and only if there exists a
$\sim'$-scheduler $\sigma'$ for $\A'$ such that
$\PrDiscl_{\sigma'}(\varphi,\A',\obs)=1$. Moreover, $\sup
\{\Prob_{\sigma}(\Acc(F)),\sigma\textrm{ $\sim$-scheduler for
  $\A$}\}=1$ if and only if $\PrDiscl(\varphi,\A',\obs)=1$.

By \cite{BaierGB12,CDGH10}, we obtain that the almost-sure disclosure
problem is undecidable for DCA (and thus for DPA), and that the almost
sure opacity is undecidable for DBA, and limit disclosure problem is
undecidable for DBA, DCA, hence for DPA that are more expressive. From
this result, we get undecidability for the general disclosure problem
and the value problems for DBA, DCA and DPA. Note that in the case of
limit disclosure, general disclosure and value problems, the
undecidability holds also when restricted to finite-memory
strategies. Indeed, undecidability results for quantitative questions
about probabilistic finite automata \cite{Rabin63,PazBook} and for
value 1 problem \cite{GO10} carry over POMDPs restricted to
finite-memory schedulers. \qed\end{proof}

We now show that, under some natural restrictions, one can recover
decidability for the almost-sure disclosure and almost-sure opacity
problems.  First, if the secret is given as a Deterministic B\"uchi
Automaton (DBA), then the almost-sure disclosure problem is decidable.
Although deterministic B\"uchi automata are strictly less expressive
than non deterministic ones, they can still be used to describe
realistic predicates. For instance, a secret which is always
recognized after a finite run would correspond to a set of runs that
reach some sink state and remain there forever. The corresponding set
of traces would be of the form $L \Sigma'^{\omega}$ for some language
$L$ of finite words and a subset $\Sigma'$ of the alphabet $\Sigma$.

\begin{theorem}\label{th:DBA}
  Given a DBA $\A_\varphi$ describing the secret,
  the almost-sure disclosure problem for POMDP is EXPTIME-complete.
\end{theorem}

\begin{proof}
  Let $\A=(Q,\Sigma,\Act,\Delta,q_0, \sim)$ be the POMDP modeling the
  system, and $\A_\varphi$ be the (complete) deterministic B\"uchi
  automaton over $Q\cup(\Act\times\Sigma)$ that recognizes the runs of
  $\A$ that are in $\varphi$.  We show how to obtain a deterministic
  automaton
  ${\A_\textrm{discl}}=(Q',Q\cup(\Act\times\Sigma),\delta,q'_0,F)$ such
  that
  $L_B(\A_\textrm{discl})=\varphi\setminus\obs^{-1}(\obs(\overline{\varphi}))$.\\
  Indeed, with a co-B\"uchi acceptance condition for $\A_\varphi$, we
  get that $L_C(\A_\varphi)=\overline{L_B(\A_\varphi)}$. Then, it is
  possible to obtain a deterministic co-B\"uchi automaton
  $\mathcal{B}$ such that
  $L_{C}(\mathcal{B})=\obs^{-1}(\obs(\overline{\varphi}))$ (recall
  that non-deterministic co-B\"uchi automata are as expressive as
  deterministic co-B\"uchi automata \cite{MiyanoHayashi84}).  Then
  $L_B(\mathcal{B})=\overline{L_C(\mathcal{B})}$, and
  $\A_{\textrm{discl}}$ is the (complete) B\"uchi automaton obtained
  by intersecting the two deterministic B\"uchi automata $\A_\varphi$
  and $\mathcal{B}$.

%
  We build a new POMDP that will jointly simulate $\A$ and
  $\A_{\textrm{discl}}$.  Since the automaton $\A_{\textrm{discl}}$
  runs over runs of $\A$, we have to make explicit the
  transitions of $\A_{\textrm{discl}}$ on states of $\A$. For that we
  introduce a copy of each state of $\A$ in the product POMDP, from
  which we will allow $\A_{\textrm{discl}}$ to take the corresponding
  transition.  Formally, we consider the product POMDP $\A\otimes
  \A_{\textrm{discl}} = (\overline{Q}\times
  Q',\Sigma\cup\{\alpha_\iota\},\Act\cup\{a_0\},\Delta',(q_0?,q'_0),\sim')$
  where $\overline{Q}=Q\cup\{q?\mid q\in Q\}$ is the set of states of
  $\A$ augmented with a copy of these states, $\alpha_\iota$ and $a_0$ are
  new symbols, and $\Delta'$ is defined as follows: for all
  $q_1,q_2\in Q$, $q'_1,q'_2\in Q'$, $\alpha\in\Act$ and $a\in\Sigma$,
\begin{align*}
\Delta'((q_1,q'_1),\alpha)(a,(q_2?,q'_2))&=\begin{cases}
\Delta(q_1,\alpha)(a,q_2) & \\
\textrm{ if $q'_2=\delta(q'_1,(\alpha,a))$}\\
0
\textrm{ otherwise}.&
\end{cases}\\
\Delta'((q_1?,q'_1),\alpha_\iota)(a_0,(q_1,q'_2))&=\begin{cases}
1 & \textrm{ if $q'_2=\delta(q'_1,q_1)$}\\
0 & \textrm{ otherwise}.
\end{cases}
\end{align*}
The new equivalence $\sim'$ is defined by: $(q_1,q'_1)\sim'
(q_2,q'_2)$ and $(q_1?,q'_1)\sim'(q_2?,q'_2)$ if $q_1\sim q_2$.  
Let $\rho'$ be a run of $\A\otimes \A_{discl}$. 
%
To define the projection of $\rho'$ on $\A$, we use the following
mapping $\Pi_\A$: for all $q\in Q$, $q'\in Q'$, $\alpha\in \Act$,
$a\in\Sigma$,
\begin{align*}
&\Pi_\A((q,q'))=q\\
&\Pi_\A((\alpha,a))=(\alpha,a)\\
&\Pi_\A((\alpha_\iota,a_0))=\Pi_\A((q?,q'))=\varepsilon
\end{align*}
which is extended to finite or infinite runs of
$\A\otimes\A_\textrm{discl}$ in the natural way.\\ Similarly, the
projection of $\rho'$ onto $\A_\textrm{discl}$ uses the following
mapping:
\begin{displaymath}
\Pi_{\textrm{discl}}:\runsf{\A\otimes\A_\textrm{discl}}
\rightarrow \runsf{\A_\textrm{discl}}
\end{displaymath}
defined by induction on the length of $\rho'$: For all $q'_1\in Q'$,
we let
$\Pi_\textrm{discl}((q_0?,q'_0)(\alpha_\iota,a_0)(q_0,q'_1))=q'_0\cdot
q_0\cdot q'_1$. Then, for all
$\rho'\in\runsf{\A\otimes\A_{\textrm{discl}}}$, for all $q_1\in Q$,
$q'_1,q'_2\in Q'$, $\alpha\in\Act$, $a\in \Sigma$, we define:
\begin{align*}
&\Pi_{\textrm{discl}}(\rho'\cdot (\alpha,a)\cdot (q_1?,q'_1)\cdot (\alpha_\iota,a_0)\cdot (q_1,q'_2))\\
&=
\Pi_{\textrm{discl}}(\rho')\cdot (\alpha,a)\cdot q'_1\cdot q_2\cdot q'_2
\end{align*}
The mapping $\Pi_{\textrm{discl}}$ is increasing, hence for $\rho'$ an
infinite run of $\A\otimes\A_{\textrm{discl}}$, we can define
$\Pi_{\textrm{discl}}(\rho')=\bigsqcup_{r\textrm{ finite prefix of
    $\rho'$}} \Pi_{\textrm{discl}}(r)$.

It is easy to see that $\rho=\Pi_\A(\rho')$ is a run of $\A$, and that
$\Pi_{\textrm{discl}}(\rho')$ is a run of $\A_{\textrm{discl}}$ over
$\rho$.  Then, $\rho\in\varphi\setminus
\obs^{-1}(\obs(\overline{\varphi}))$ if and only if $\rho\in
L_B(\A_\textrm{discl})$, if and only if
$\Pi_{\textrm{discl}}(\rho')\in\Buchi{F}$ if and only if
$\rho'\in\Buchi{\overline{Q}\times F}$.

Let $\sigma'$ be a $\sim'$-scheduler of
$\A\otimes\A_{\textrm{discl}}$, and let $\rho$ be a finite run of
$\A$. Observe that there is a unique run
$\rho'\in\runsf{\A\otimes\A_{\textrm{discl}}}$ such that
$\Pi_\A(\rho')=\rho$. We then let $\sigma(\rho)=\sigma'(\rho')$, which
is clearly a $\sim$-scheduler for $\A$.
Moreover, for all finite runs $\rho$ of $\A$, we have $\Prob_{\sigma}(C_{\rho})=\Prob_{\sigma'}(C_{\rho'})$. Hence $\Prob_\sigma(\varphi\setminus \obs^{-1}(\obs(\overline{\varphi}))) = \Prob_{\sigma'}(\Buchi{\overline{Q}\times F})$.\\
Conversely, let $\sigma$ be a $\sim$-scheduler of $\A$. We define a
$\sim'$-scheduler $\sigma'$ as follows. For
$\rho'\in\runsf{\A\otimes\A_{\textrm{discl}}}$, for all $q\in Q$,
$q'\in Q'$,
\begin{align*}
&\sigma'(\rho'\cdot (q?,q))=(\alpha_\iota\mapsto 1)\\
&\sigma'(\rho'\cdot (q,q'))=\sigma(\Pi_\A(\rho'\cdot (q,q')))
\end{align*}
In that case again, $\Prob_\sigma(\Pi_\A(\rho'))=\Prob_{\sigma'}(\rho')$, so $\Prob_\sigma(\varphi\setminus \obs^{-1}(\obs(\overline{\varphi}))) = \Prob_{\sigma'}(\Buchi{\overline{Q}\times F})$.

Now, the almost-sure disclosure problem is equivalent to deciding
whether there is a $\sim'$-scheduler $\sigma'$ for
$\A\otimes\A_{\textrm{discl}}$ such that
$\Prob_{\sigma'}(\Buchi{\overline{Q}\times F})=1$. From
\cite{BaierGB12, CDR07, CDH-mfcs10}, this last problem is in EXPTIME.
To solve the problem on a given POMDP, one builds an MDP in which
each state is enriched with the belief of the scheduler at this point,
hence with a size exponentially larger than the original model.  A
naive application of this algorithm to the POMDP $\A\otimes
\A_{\textrm{discl}}$ would yield a POMDP of size exponentially larger
than the original $\A$ and $\A_\varphi$, hence would provide an
algorithm in 2-EXPTIME. We then need a more careful and less costly
construction: it consists in computing the belief only on the POMDP $\A$
part, and not on the component coming from $\A_{\textrm{discl}}$,
which is simply a deterministic automaton.  Hence, the obtained MDP is
only exponential in the size of $\A$ and $\A_\varphi$, and the overall
algorithm is in EXPTIME.

Moreover, proof of Theorem~\ref{th:undecidable} provides a reduction
from qualitative problems on POMDP to almost-sure opacity and almost
sure disclosure problems.  Given a run $\rho$, we let $\Appear{\rho}$
be the set of states appearing (at least once) in the run, and
consider the acceptance condition $\reach{F}$ defined, for $F\subseteq
Q$, by $\reach{F}=\{\rho\in Q^\omega\mid \Appear{\rho}\cap
F\neq\emptyset\}$.  Then, we have shown that given a POMDP $\A$, and a
set of states $F$, one can build a POMDP $\A'$ (which is the POMDP of
Figure~\ref{fig:undec}, in which the set $F_1$ is made absorbing), an
observation function $\obs$, and a secret $\varphi$ given by a DBA,
such that there exists a $\sim$-scheduler for $\A$ such that
$\Prob_\sigma(\reach{F})=1$ if and only if there exists a
$\sim'$-scheduler $\sigma'$ for $\A'$ such that
$\PrDiscl_{\sigma'}(\A',\obs,\varphi)=1$. The EXPTIME-hardness for our
problem follows from the EXPTIME-hardness of the almost-sure problem
for POMDP with reachability conditions~\cite{CDH-mfcs10}.  \qed
\end{proof}

Finally, we show that if we restrict to finite-memory schedulers,
then both the almost-sure disclosure and almost-sure opacity problems
become decidable for secrets given as DPA. Since finite-memory schedulers
are the only schedulers of practical interest, and DPA allow to describe any
$\omega$-regular predicate, this restriction is of great interest.

\begin{theorem}
  Given a POMDP $\A$, a morphism $\obs$ as observation function, 
and a secret given as a DPA, the finite-memory
  almost-sure opacity problem and the finite-memory almost-sure
  disclosure problem are EXPTIME-complete.
\end{theorem}

\begin{proof}
The proof follows the same lines than the proof of Theorem~\ref{th:DBA}.
Given a POMDP $\A=(Q,\Sigma,\Act,\Delta,q_0,\sim)$ modeling the system
and a DPA $\A_\varphi$ describing the secret $\varphi$, one can obtained a DPA 
$\A_{\textrm{discl}}=(Q',Q\cup(\Act\times\Sigma),\delta,q'_0,F)$ such that 
$L_P(\A_\textrm{discl})=\varphi\setminus\obs^{-1}(\obs(\overline{\varphi}))$, 
since this language is $\omega$-regular.

As in the previous proof, we build a new POMDP as a product of $\A$
and $\A_\textrm{discl}$,
$\A\otimes\A_\textrm{discl}=(\overline{Q}\times Q',
\Sigma\cup\{\alpha_\iota\}, \Delta',(q_0?,q'_0),\sim')$. If
$F:Q'\rightarrow \{1,\cdots,k\}$, we let $F':\overline{Q}\times
Q'\rightarrow\{1,\dots,k\}$, where, for all $q\in\overline{Q},q'\in
Q'$, $F'(q,q')=F(q')$.  Then, the finite-memory almost-sure disclosure problem is
equivalent to deciding whether there is a finite-memory $\sim'$-scheduler $\sigma'$
for $\A\otimes\A_{\textrm{discl}}$ such that
$\Prob_{\sigma'}(\parity{F'})=1$, and the finite-memory almost-sure opacity problem
is equivalent to deciding whether there is a finite-memory scheduler $\sigma'$ for
$\A\otimes\A_{\textrm{discl}}$ such that
$\Prob_{\sigma'}(\parity{F'})=0$.  From \cite{CCT13}, when restricting to finite-memory
schedulers, these two
problems are in EXPTIME. As in the proof of Theorem~\ref{th:DBA}, to
maintain the procedure within exponential time, the powerset
construction on the POMDP will only be made on the $\A$ component of
the product.

Also, the proof of EXPTIME-hardness follows the same lines than the
proof of Theorem~\ref{th:DBA}.  \qed
%

\end{proof}

\section{Conclusion}\label{sec:conc}
Extending the definition of probabilistic opacity to MDPs (with
infinite runs), we solve decidability questions raised
in~\cite{berard13}. More elaborate measures could be studied in this
context, and are left for future work. Another interesting issue would
be to investigate \emph{disclosure before some given delay}, either as
a number of steps in the spirit
of~\cite{DBLP:journals/tase/SabooriH11} for discrete event systems, or
for probabilistic timed system with an explicit time bound. In the
latter case, decidability results could be obtained by combining our
results with recent ones like~\cite{DBLP:conf/atva/Brihaye0GORW13}.

\medskip \noindent \textbf{Acknowledgements.} We thank anonymous
referees for their comments and suggestions.  The research was partly
supported by Austrian Science Fund (FWF) Grant No P 23499- N23, FWF
NFN Grant No S11407-N23, ERC Start grant (279307: Graph Games),
Microsoft faculty fellows award, Coop\'eration France-Qu\'ebec,
Service Coop\'eration et Action Culturelle 2012/26/SCAC, and project ImpRo 
ANR-2010-BLAN-0317.


\medskip \noindent
\textbf{References}
\begin{thebibliography}{10}

\bibitem{BaierGB12}
Christel Baier, Marcus Gr{\"o}{\ss}er, and Nathalie Bertrand.
\newblock Probabilistic $\omega$-automata.
\newblock {\em J. ACM}, 59(1):1, 2012.

\bibitem{berard10}
B{\'e}atrice B{\'e}rard, John Mullins, and Mathieu Sassolas.
\newblock Quantifying opacity.
\newblock In Gianfranco Ciardo and Roberto Segala, editors, {\em Proceedings of
  the 7th International Conference on Quantitative Evaluation of Systems
  (QEST'10)}, pages 263--272. IEEE Computer Society, September 2010.

\bibitem{berard13}
B{\'e}atrice B{\'e}rard, John Mullins, and Mathieu Sassolas.
\newblock Quantifying opacity.
\newblock {\em CoRR}, abs/1301.6799, 2013.
\newblock extended version.

\bibitem{BillingsleyMeasure95}
Patrick Billingsley.
\newblock {\em {Probability and Measure}}.
\newblock Wiley, New York, NY, 3rd edition, 1995.

\bibitem{boreale11b}
Michele Boreale, Francesca Pampaloni, and Michela Paolini.
\newblock Quantitative information flow, with a view.
\newblock In Vijay Atluri and Claudia D\'{\i}az, editors, {\em Proc. of 16th
  European Symposium on Research in Computer Security (ESORICS 2011)}, volume
  6879 of {\em Lecture Notes in Computer Science}, pages 588--606. Springer,
  2011.

\bibitem{DBLP:conf/atva/Brihaye0GORW13}
Thomas Brihaye, Laurent Doyen, Gilles Geeraerts, Jo{\"e}l Ouaknine,
  Jean-Fran\c{c}ois Raskin, and James Worrell.
\newblock Time-bounded reachability for monotonic hybrid automata: Complexity
  and fixed points.
\newblock In Dang~Van Hung and Mizuhito Ogawa, editors, {\em Proc. of 11th
  International Symposium on Automated Technology for Verification and
  Analysis, ATVA 2013}, volume 8172 of {\em Lecture Notes in Computer Science},
  pages 55--70. Springer, 2013.

\bibitem{bryans08}
Jeremy~W. Bryans, Maciej Koutny, Laurent Mazar\'e, and Peter Y.~A. Ryan.
\newblock Opacity generalised to transition systems.
\newblock {\em Intl. Jour. of Information Security}, 7(6):421--435, 2008.

\bibitem{DBLP:conf/tgc/BryansKM12}
Jeremy~W. Bryans, Maciej Koutny, and Chunyan Mu.
\newblock Towards quantitative analysis of opacity.
\newblock In Catuscia Palamidessi and Mark~Dermot Ryan, editors, {\em Proc. 7th
  Int. Symp. on Trustworthy Global Computing (TGC'12), Revised Selected
  Papers}, volume 8191 of {\em Lecture Notes in Computer Science}, pages
  145--163. Springer, 2013.

\bibitem{dubreil12}
Franck Cassez, Jeremy Dubreil, and Herv\'e Marchand.
\newblock {Synthesis of opaque systems with static and dynamic masks}.
\newblock {\em Formal Methods in System design}, 40(1):88 --115, 2012.

\bibitem{CCT13}
Krishnendu Chatterjee, Martin Chmelik, and Mathieu Tracol.
\newblock What is decidable about partially observable Markov decision
  processes with omega-regular objectives.
\newblock In {\em CSL}, pages 165--180, 2013.

\bibitem{CDGH10}
Krishnendu Chatterjee, Laurent Doyen, Hugo Gimbert, and Thomas~A. Henzinger.
\newblock Randomness for free.
\newblock In {\em Proceedings of MFCS 2010: Mathematical Foundations of
  Computer Science}, Lecture Notes in Computer Science 6281, pages 246--257.
  Springer-Verlag, 2010.

\bibitem{CDH-mfcs10}
Krishnendu Chatterjee, Laurent Doyen, and {\relax Th}omas~A. Henzinger.
\newblock Qualitative analysis of partially-observable {M}arkov decision
  processes.
\newblock In Petr Hlin{\v e}n{\'y} and Anton{\'\i}n Ku{\v c}era, editors, {\em
  {P}roceedings of the 35th {I}nternational {S}ymposium on {M}athematical
  {F}oundations of {C}omputer {S}cience ({MFCS}'10)}, volume 6281 of {\em
  Lecture Notes in Computer Science}, pages 258--269, Brno, Czech Republic,
  August 2010. Springer.

\bibitem{CDR07}
Krishnendu Chatterjee, Laurent Doyen, Thomas~A. Henzinger, and
  Jean-Fran\c{c}ois Raskin.
\newblock Algorithms for omega-regular games with imperfect information.
\newblock {\em Logical Methods in Computer Science}, 3(3), 2007.

\bibitem{CH11}
Krishnendu Chatterjee and Monika Henzinger.
\newblock Faster and dynamic algorithms for maximal end-component decomposition
  and related graph problems in probabilistic verification.
\newblock In {\em SODA}, pages 1318--1336, 2011.

\bibitem{CJH04}
Krishnendu Chatterjee, Marcin Jurdzinski, and Thomas~A. Henzinger.
\newblock Quantitative stochastic parity games.
\newblock In {\em SODA}, pages 121--130, 2004.

\bibitem{chaum88}
David Chaum.
\newblock The dining cryptographers problem: unconditional sender and recipient
  untraceability.
\newblock {\em Journal of Cryptology}, 1:65--75, 1988.

\bibitem{courcoubetis95}
Costas Courcoubetis and Mihalis Yannakakis.
\newblock The complexity of probabilistic verification.
\newblock {\em {Journal of the ACM}}, 42(4):857--907, 1995.

\bibitem{dubreil10}
Jeremy Dubreil, Philippe Darondeau, and Herv\'e Marchand.
\newblock {Supervisory Control for Opacity}.
\newblock {\em IEEE Transactions on Automatic Control}, 55(5):1089 --1100, may
  2010.

\bibitem{GO10}
Hugo Gimbert and Youssouf Oualhadj.
\newblock Probabilistic automata on finite words: Decidability and
  undecidability results.
\newblock In Samson Abramsky, Cyril Gavoille, Claude Kirchner, Friedhelm Meyer
  auf~der Heide, and Paul~G. Spirakis, editors, {\em Proceedings of ICALP
  2010}, volume 6199 of {\em Lecture Notes in Computer Science}, pages
  527--538. Springer, 2010.

\bibitem{goguen82}
Joseph~A. Goguen and Jos{\'e} Meseguer.
\newblock Security policy and security models.
\newblock In {\em Proc. of {IEEE} {S}ymposium on {S}ecurity and {P}rivacy},
  pages 11--20. {IEEE} Computer Society Press, 1982.

\bibitem{lakhnech05}
Yassine Lakhnech and Laurent Mazar\'e.
\newblock Probabilistic opacity for a passive adversary and its application to
  {C}haum's voting scheme.
\newblock Technical Report~4, Verimag, 2 2005.

\bibitem{mantel2009}
Heiko Mantel and Henning Sudbrock.
\newblock Information-theoretic modeling and analysis of interrupt-related
  covert channels.
\newblock In P.~Degano, J.~Guttman, and F.~Martinelli, editors, {\em
  Proceedings of the Workshop on Formal Aspects in Security and Trust, FAST
  2008}, Springer, LNCS 5491, pages 67--81, 2009.

\bibitem{pinchinat}
Bastien Maubert, Sophie Pinchinat, and Laura Bozzelli.
\newblock Opacity issues in games with imperfect information.
\newblock In {\em 2nd Int. Symp. on Games, Automata, Logics and Formal
  Verification (GandALF'11)}, pages 87--101, 2011.

\bibitem{mazare05}
Laurent Mazar\'e.
\newblock Decidability of opacity with non-atomic keys.
\newblock In {\em Proc. 2nd {W}orkshop on {F}ormal {A}spects in {S}ecurity and
  {T}rust ({FAST}'04)}, volume 173 of {\em Intl. Federation for Information
  Processing}, pages 71--84. Springer, 2005.

\bibitem{millen87}
Jonathan~K. Millen.
\newblock {Covert Channel Capacity}.
\newblock In {\em Proc. of IEEE Symposium on Research in Computer Security and
  Privacy}, pages 144--161, 1987.

\bibitem{MiyanoHayashi84}
Satoru Miyano and Takeshi Hayashi.
\newblock Alternating finite automata on omega-words.
\newblock {\em Theor. Comput. Sci.}, 32:321--330, 1984.


\bibitem{PazBook}
A.~Paz.
\newblock {\em Introduction to probabilistic automata (Computer science and
  applied mathematics)}.
\newblock Academic Press, 1971.

\bibitem{Rabin63}
Michael~O. Rabin.
\newblock Probabilistic automata.
\newblock {\em Information and Control}, 6(3):230--245, 1963.

\bibitem{reiter98}
Michael~K. Reiter and Aviel~D. Rubin.
\newblock Crowds: anonymity for web transactions.
\newblock {\em ACM Transactions on Information and System Security},
  1(1):66--92, 1998.

\bibitem{DBLP:journals/tase/SabooriH11}
Anooshiravan Saboori and Christoforos~N. Hadjicostis.
\newblock Verification of k-step opacity and analysis of its complexity.
\newblock {\em IEEE T. Automation Science and Engineering}, 8(3):549--559,
  2011.

\bibitem{DBLP:journals/tac/SabooriH14}
Anooshiravan Saboori and Christoforos~N. Hadjicostis.
\newblock Current-state opacity formulations in probabilistic finite automata.
\newblock {\em IEEE Trans. Automat. Contr.}, 59(1):120--133, 2014.

\bibitem{Thomas97}
Wolfgang Thomas.
\newblock Languages, automata, and logic.
\newblock In {\em Handbook of Formal Languages}, pages 389--455. Springer,
  1997.

\bibitem{Vardi85}
Moshe~Y. Vardi.
\newblock Automatic verification of probabilistic concurrent finite-state
  programs.
\newblock In {\em Proceedings of 26th Annual Symposium on Foundations of
  Computer Science (FOCS)}, pages 327--338. IEEE Computer Society, 1985.

\bibitem{wittbold90}
John~T. Wittbold and Dale~M. Johnson.
\newblock Information flow in nondeterministic systems.
\newblock In {\em Proc. of IEEE Symposium on Research in Computer Security and
  Privacy}, pages 144--161, 1990.

\end{thebibliography}

\end{document}